\documentclass[letterpaper, 10 pt, conference]{ieeeconf}  

\IEEEoverridecommandlockouts                              
 \usepackage{amssymb}
\usepackage{amsmath}
\usepackage{amsfonts}
\usepackage{graphicx}
\usepackage{caption}
\usepackage{enumerate}
\usepackage{multirow}
\usepackage{mathrsfs}
\let\labelindent\relax	
\usepackage[inline]{enumitem}
\usepackage{euscript}
\usepackage{breqn}
\usepackage{cite}

\newcommand{\NX}{{n_{\mathrm{x}}}}
\newcommand{\NY}{{n_{\mathrm{y}}}}
\newcommand{\NU}{{n_{\mathrm{u}}}}
\newcommand{\NP}{{n_{\mathrm{p}}}}
\newcommand{\NO}{{n_{\mathrm{o}}}}
\newcommand{\NR}{{n_{\mathrm{r}}}}

\newcommand{\NB}{{n_{\mathrm{b}}}}
\newcommand{\NC}{{n_{\mathrm{c}}}}
\newcommand{\NBh}{{\hat{n}_{\mathrm{b}}}}
\newcommand{\NCh}{{\hat{n}_{\mathrm{c}}}}

\newcommand{\sU}{\mathbb{U}}
\newcommand{\sP}{\mathbb{P}}
\newcommand{\sX}{\mathbb{X}}
\newcommand{\sY}{\mathbb{Y}}
\newcommand{\sZ}{\mathbb{Z}}
\newcommand{\sI}{\mathbb{I}}

\newcommand{\Ru}{\mathbb{R}^{n_{\mathrm{u}}}}
\newcommand{\Rp}{\mathbb{R}^{n_{\mathrm{p}}}}
\newcommand{\Rx}{\mathbb{R}^{n_{\mathrm{x}}}}
\newcommand{\Ry}{\mathbb{R}^{n_{\mathrm{y}}}}

\newcommand{\hank}{\mathcal{H}}

\newcommand{\expct}{\mathbb{E}}

\newcommand{\selO}{\nu}
\newcommand{\selR}{\varsigma}

\newcommand{\SNR}{\mathrm{SNR}}

\newcommand{\BFR}{\mathrm{BFR}}

\newcommand{\Tr}{{\mathrm{Tr}}}

\newcommand{\Afnc}{\mathcal{A}}
\newcommand{\Bfnc}{\mathcal{B}}
\newcommand{\Cfnc}{\mathcal{C}}
\newcommand{\Dfnc}{\mathcal{D}}
\newcommand{\Kfnc}{\mathcal{K}}

\newcommand{\Bflt}{\EuScript{B}}
\newcommand{\Cflt}{\EuScript{C}}
\newcommand{\Bflth}{\hat{\EuScript{B}}}
\newcommand{\Cflth}{\hat{\EuScript{C}}}

\newcommand{\Bpar}{\mathtt{g}}
\newcommand{\Cpar}{\mathtt{h}}
\newcommand{\Bparh}{\hat{\mathtt{g}}}
\newcommand{\Cparh}{\hat{\mathtt{h}}}

\newcommand{\numstr}[1]{\#(#1)}

\newcommand{\Dat}{\EuScript{D}_N}
\newcommand{\Dval}{\EuScript{D}_{\mathrm{val}}}

\newcommand{\Gamfnc}{\hat\Gamma}

\newtheorem{thm}{Theorem}
\newtheorem{lem}[thm]{Lemma}

\title{\LARGE \bf Alternative Form of Predictor Based Identification of \\ LPV-SS Models with Innovation Noise
}

\author{Pepijn Cox$^\dag$ and Roland T\'{o}th$^\dag$  
\thanks{$^\dag$ Pepijn Cox and Roland T\'{o}th are with the Control Systems Group, Department of Electrical Engineering, Eindhoven University of Technology, P.O. Box 513, 5600 MB Eindhoven, The Netherlands,
        {\tt\small \{p.b.cox,r.toth\}@tue.nl}}%
}
\begin{document}

\maketitle
\thispagestyle{empty}
\pagestyle{empty}
\begin{abstract}
In this paper, we present an approach to identify linear parameter-varying (LPV) systems with a state-space (SS) model structure in an innovation form where the coefficient functions have static and affine dependency on the scheduling signal. With this scheme, the curse of dimensionality problem is reduced, compared to existing predictor based LPV subspace identification schemes. The investigated LPV-SS model is reformulated into an equivalent impulse response form, which turns out to be a moving average with exogenous inputs (MAX) system. The Markov coefficient functions of the LPV-MAX representation are multi-linear in the scheduling signal and its time-shifts, contrary to the predictor based schemes where the corresponding LPV auto-regressive with exogenous inputs system is multi-quadratic in the scheduling signal and its time-shifts. In this paper, we will prove that under certain conditions on the input and scheduling signals, the $\ell_2$ loss function of the one-step-ahead prediction error for the LPV-MAX model has only one unique minimum, corresponding to the original underlying system. Hence, identifying the LPV-MAX model in the prediction error minimization framework will be consistent and unbiased. The LPV-SS model is realized by applying an efficient basis reduced Ho-Kalman realization on the identified LPV-MAX model. The performance of the proposed scheme is assessed on a Monte Carlo simulation study.
\end{abstract}
\vspace{-1.5mm}
\section{Introduction}

In recent years, identification of \textit{linear parameter-varying state-space} (LPV-SS) models has received considerable attention, e.g.,~\cite{Wingerden2009a,Wills2011,Verdult2003,Toth2012}, resulting in a succesful extention of \textit{subspace identification} (SID) and \textit{prediction error minimization} (PEM) methods to the LPV case. The LPV-SS identification framework has been applied to data-driven modelling of wind turbines~\cite{Wingerden2009}, forced Lorenz attractors~\cite{Larimore2015}, power management of Web service systems~\cite{Tanelli2008}, or to capture traffic flow~\cite{Luspay2009}, to mention a few. A popular choice is to capture the underlying system as a discrete-time LPV-SS representation with an innovation noise model: \vspace{-2mm}
\begin{subequations} \label{eq:SSrep}
\begin{alignat}{3}
  qx &= \Afnc(p)&x&+\Bfnc(p)&&u+\Kfnc(p)e,  \label{eq:SSrepState}  \\
  y   &= \Cfnc(p)&x&+\Dfnc(p)&&u+e, \label{eq:SSrepOut}
\end{alignat}
\end{subequations}\vskip -1.5mm \noindent
where $x:\sZ\rightarrow\sX=\Rx$ is the state variable, $y:\sZ\rightarrow\sY=\Ry$ is the measured output signal, $u:\sZ\rightarrow\sU=\Ru$ denotes the input signal, $p:\sZ\rightarrow\sP \subseteq\Rp$ is the scheduling variable, $q$ is the forward time-shift operator, e.g., $qx(t)=x(t+1)$ where $t\in\mathbb{Z}$ is the discrete time, $e:\sZ\rightarrow\Ry$ is a sampled path of a zero-mean i.i.d. stationary noise processes with Gaussian distribution, i.e., $e(t)\sim \mathcal{N}(0,\Sigma_{\mathrm{e}})$ with a nonsingular covariance $\Sigma_{\mathrm{e}}\in\mathbb{R}^{\NY\times\NY}$. The matrix functions $\Afnc(\cdot),...,\Kfnc(\cdot)$ defining the SS representation \eqref{eq:SSrep} are usually taken to be affine combinations of $p$, to coincide with the majority of LPV control synthesis methods, e.g.,~\cite{Mohammadpour2012}. Hence, these matrix functions are defined as \vspace{-1.5mm}
\begin{equation}\label{eq:sysMatrices}
\begin{aligned} 
 \Afnc(p)&=A_0+\hspace{-1mm}\sum_{i=1}^{\NP} A_i p_i,
\end{aligned}
\end{equation} \vskip -1mm \noindent
where $\Bfnc(p),\ldots,\Kfnc(p)$ are equivalently parametrized as~\eqref{eq:sysMatrices} with $\{A_i,B_i,C_i,D_i,K_i\}_{i=0}^{\NP}$ being constant matrices of appropriate dimensions. 

Identification of LPV-SS models can be formulated in the PEM framework, where the PEM methods can roughly be categorized as:
\begin{enumerate*}[label=\roman*)]
	\item methodologies which assume that full state measurements are available, e.g.,~\cite{Nemani1995};
	\item grey box schemes, where only a small subset of all parameters is estimated, e.g.,~\cite{Gaspar2007};
	\item set-membership approaches, e.g.,~\cite{Novara2011}; or
	\item direct PEM methods using gradient based methodologies, e.g.,~\cite{Wills2011,Verdult2003}.
\end{enumerate*}
However, assuming that a full measurement of the state or detailed model of the system is available, is unrealistic in many applications. Therefore, only set-membership and direct PEM methods classify as `black-box' identification schemes. In general, the set-membership approaches have a significant higher computational load compared to the direct PEM methods, as the set-membership schemes estimate a feasible model set for which only heuristic techniques exist. In addition, these schemes have often unrealistic noise assumptions and consistency of the overall scheme is difficult to show. For the direct PEM methods, the nonlinear and nonunique optimization problem is usually solved by applying gradient-based search strategies or an expectation maximization strategy. Regardless of the strategy, direct PEM methods are solved in an iterative way, are prone to local minima, and their convergence depends heavily on a proper initial guess.  

On the contrary, SID methods use convex optimization, to identify a specific LPV \textit{input-output} (IO) structure, from which an LPV-SS model is realized by using matrix decomposition methods. Subspace schemes have extensively been applied in the \textit{linear time invariant} (LTI) case and extensions to the LPV framework can be found, e.g.,~\cite{Toth2012,Larimore2013,Wingerden2009a,Santos2007}. Unfortunately, SID methods usually depend on approximations to get a convex problem, and, as a consequence, in the LPV case, these approaches suffer heavily from the curse of dimensionality and/or result in ill-conditioned estimation problems with high parameter variances. 

Based on these considerations, efficient estimation of LPV-SS models remains a central problem to be solved.

In this paper, we would like to reduce the dimensionality problems associated with the predictor based SID methods by restating the LPV-IO identification setting to avoid multi-quadratic dependency on the scheduling signal and its time-shifts. First, the LPV-SS representation~\eqref{eq:SSrep} is formulated in terms of its \textit{impulse response representation} (IIR). As Section~\ref{sec:impulseResponseRep} highlights, the corresponding IIR turns out to be a \textit{moving average with exogenous inputs} (MAX) representation, where the Markov coefficient functions are multi-linear in the scheduling signal and its time-shifts. In Section~\ref{sec:MaxIdent}, the identification of the multivariable MAX model is given. We will proof that minimizing the $\ell_2$ loss function of the one-step-ahead prediction error has only one unique minimum under some mild conditions on the input and scheduling signals, which is the main contribution of this paper. To this end, we extent the uniqueness proof of the LTI multivariable moving average model~\cite{Stoica1982} to the LPV case, which is essential in showing uniqueness of the overall LPV-MAX PEM identification method. Hence, applying a pseudo linear regression to identify the LPV-MAX model results in a consistent and unbiased estimate. After identifying the LPV-MAX model, the LPV-SS model is estimated by utilizing the bases reduced Ho-Kalman realization scheme of~\cite{Cox2015}, which, in this paper, is modified to also realize the noise model. In Section~\ref{sec:example}, the performance of the identification method is assessed by a Monte-Carlo study, followed by some conclusions in Section~\ref{sec:conclusion}.

\vspace{-1.5mm}
\section{Impulse Response Representation} \label{sec:impulseResponseRep} \vspace{-1.5mm}

To be able to solve our identification problem, we use an alternative representation of~\eqref{eq:SSrep}:
\begin{lem}[IIR representation~\cite{Toth2010a}] 
\label{lmn:IIR}
Any asymptotically stable\footnote{An LPV system, represented in terms of~\eqref{eq:SSrep}, is called asymptotically stable in the deterministic sense, if, for all trajectories of $(u(t),e(t),p(t),y(t))$ satisfying~\eqref{eq:SSrep}, with $u(t)=0$, $e(t)=0$ for $t\geq0$ and $p(t)\in\sP$, it holds that $\lim_{t\rightarrow\infty}\vert y(t)\vert = 0$.} LPV system given in terms of representation \eqref{eq:SSrep} has a convergent series expansion in terms of the pulse-basis $\{q^{-i}\}_{i=0}^\infty$ given by \vspace{-2mm}
	\begin{equation} \label{eq:IIR}
		y = \underbrace{\sum_{i=0}^\infty(g_i\diamond p)q^{-i}u}_{\mbox{\footnotesize process model}} +\underbrace{\sum_{j=0}^\infty(h_j\diamond p)q^{-j}e}_{\mbox{\footnotesize noise model}},
	\end{equation} \vskip -1.5mm \noindent
	where $g_i\in \mathscr{R}^{\NY\times\NU}$ and $h_j\in \mathscr{R}^{\NY\times\NY}$ are the expansion coefficient functions, i.e., \emph{Markov coefficients}, of the process and noise dynamics, respectively. The ring of all real polynomial functions with finite dimensional domain is defined by $\mathscr{R}$, the operator $\diamond:(\mathscr{R},\sP^\sZ)\rightarrow(\mathbb{R}^{\NY\times\NU})^\sZ$ denotes $(g_i\diamond p)=g_i(p(t),\ldots,p(t-\tau))$ with $\tau\in\mathbb{Z}$. \hfill $\square$
\end{lem}
The IIR coefficients $\{g_i,h_i\}_{i=0}^\infty$ of~\eqref{eq:IIR} are given by \vspace{-1.5mm}
\begin{multline}  \label{eq:MarkovParam} 
y = \underbrace{\Dfnc(p)}_{g_0\diamond p}u + \underbrace{\Cfnc(p)\Bfnc(q^{-1}p)}_{g_1\diamond p}q^{-1}u+ \\
 \underbrace{\Cfnc(p)\Afnc(q^{-1}p)\Bfnc(q^{-2}p)}_{g_2\diamond p}q^{-2}u + \ldots+ \underbrace{I}_{h_0\diamond p} e + \\\
 \underbrace{\Cfnc(p)\Kfnc(q^{-1}p)}_{h_1\diamond p}q^{-1}e \!+\! \underbrace{\Cfnc(p)\Afnc(q^{-1}p)\Kfnc(q^{-2}p)}_{h_2\diamond p}q^{-2}e \!+\! \ldots\!
\end{multline} \vskip -1.5mm \noindent
where $g_i,h_i$ converges to the zero function in de $\ell_\infty$ sense as $i\rightarrow\infty$. The Markov coefficients of the process part can be written as\vspace{-2mm}
\begin{multline} 
g_m\diamond p =\Cfnc(p)\Afnc(q^{-1}p)\cdots \Afnc(q^{-(m-1)}p)\Bfnc(q^{-m}p) = \\ 
\!\sum_{i=0}^\NP\!\sum_{j=0}^\NP\!\cdots\!\!\sum_{k=0}^\NP\!\sum_{l=0}^\NP C_iA_j\cdots A_kB_lp_i (q^{-1}p_j\!)\!\cdots\!(q^{-m}p_l\!), \label{eq:subMarkovParam}
\end{multline} \vskip -1.5mm \noindent
and the Markov coefficients of the noise part $h_m$ are similar to~\eqref{eq:subMarkovParam}, however, $\Bfnc(q^{-m}p)$ is exchanged for $\Kfnc(q^{-m}p)$. The individual products $C_iA_j\cdots A_kB_l$ or $C_iA_j\cdots A_kK_l$ are called the \textit{sub-Markov parameters} of the process model and noise model, respectively. 
Due to the convergence of $g_i$ and $h_j$, it is often sufficient to truncate \eqref{eq:IIR}: \vspace{-2mm} 
\begin{equation} \label{eq:FIRmodels}
	y \approx \sum_{i=0}^\NB(g_i\diamond p)q^{-i}u + \sum_{j=0}^\NC(h_j\diamond p)q^{-j}e,
\end{equation}\vskip -1.5mm \noindent
where $\NB>0$ and $\NC>0$ are the order of the resulting \textit{finite impulse response} (FIR) models for the process and noise part.

The sub-Markov parameters in~\eqref{eq:subMarkovParam} have a multi-linear dependency on the elements of $p$ and its time-shifts, contrary to the LPV \emph{auto-regressive with exogenous inputs} (ARX) formulation of other predictor based subspace schemes, e.g.~\cite{Wingerden2009a}, which have multi-quadratic dependency. The LPV-ARX model\footnote{The LPV-ARX representation is found by substituting $e$ of~\eqref{eq:SSrepOut} into~\eqref{eq:SSrepState} and, by using this modified state equation, writing out~\eqref{eq:SSrepOut} as: \vspace{-3.5mm}
\begin{multline} \label{eq-ftn:LPV-ARX}
y = \Dfnc(p)u+\sum_{i=1}^\infty \Cfnc(p) \Big[\prod_{j=1}^{i-1} \tilde\Afnc(q^{-j}p)\Big]\tilde\Bfnc(q^{-i}p)q^{-i}u \\[-2mm] 
+\sum_{i=1}^\infty \Cfnc(p) \Big[\prod_{j=1}^{i-1} \tilde\Afnc(q^{-j}p)\Big]\Kfnc(q^{-i}p)q^{-i}y + e
\end{multline} \vskip -2mm \noindent
where $\tilde\Afnc(p)\!=\!\Afnc(p)\mbox{\small-}\Kfnc(p)\Cfnc(p)$, $\tilde\Bfnc(p)\!=\!\Bfnc(p)\mbox{\small-}\Kfnc(p)\Dfnc(p)$, and $\prod_{j=1}^{0}\!\tilde\Afnc\!=\!I$.} uses $\tilde \Afnc(p)=\Afnc(p)-\Kfnc(p)\Cfnc(p)$, therefore, increasing the complexity of the LPV-IO model to be identified. To see this, substitute $\Afnc(p)=\sum A_i p_i$ by $\tilde\Afnc(p)=\sum\sum A_i p_i-K_iC_j p_ip_j$ (with $p_0=1$) in~\eqref{eq:subMarkovParam} and, similarly, $\Bfnc$ by $\tilde\Bfnc$. To compare the additional parameters, truncate the LPV-ARX model~\eqref{eq-ftn:LPV-ARX} by orders $n_{\mathrm{a}}$ and $n_{\mathrm{d}}$ for the $u$ and $y$ polynomial, respectively; then the LPV-ARX model~\eqref{eq-ftn:LPV-ARX} has $\NY(\NU\sum_{i=0}^{n_{\mathrm{a}}}(1+\NP)^{2i+1}+\NY\sum_{j=1}^{n_{\mathrm{d}}}(1+\NP)^{2j})$ parameters and the LPV model~\eqref{eq:FIRmodels} has $\NY (\NU\sum_{i=1}^{\NB+1}(1+\NP)^i+\NY\sum_{j=2}^{\NC+1}(1+\NP)^j)$. For example, if $\NU\!\!=\!\!\NY\!\!=\!\!\NP\!\!=\!\!n_{\mathrm{a}}\!\!=\!\!\NB\!\!=\!\!\NC\!\!=\!\!n_{\mathrm{d}}\!\!=\!\!2$ then the LPV-ARX model has $1452$ parameters and the proposed LPV model has $300$, hence, a significant reduction is achieved. Therefore, a common assumption by other predictor based subspace schemes is to take either $\Kfnc$ or $\Cfnc$,$\Dfnc$ to be constant matrices, avoiding the additional modelling complexity. However, by identifying~\eqref{eq:MarkovParam}, we can keep $\Kfnc$, $\Cfnc$, and $\Dfnc$ to be parameter dependent.

\section{Identification of a MAX model} \label{sec:MaxIdent} 

\subsection{Problem Setting} \label{sec:MaxProblemSet} 

The Markov coefficients~\eqref{eq:MarkovParam} of the impulse response~\eqref{eq:FIRmodels} representing the LPV-SS system~\eqref{eq:SSrep} are functions in the scheduling signal and its time shifts, i.e.,~\eqref{eq:subMarkovParam}. Hence~\eqref{eq:FIRmodels} is a \emph{moving average with exogenous inputs} (MAX) system: \vspace{-1mm}
\begin{equation}
y = \Bflt(p,q^{-1})u + \Cflt(p,q^{-1})e, \label{eq:MAXsys}
\end{equation} \vskip -1.5mm \noindent
with the process and noise filter given by \vspace{-1mm}
\begin{equation}
\Bflt(p,q^{-1})\!=\!\sum_{i=0}^\NB\Bflt_i(p)q^{-i},~~~\Cflt(p,q^{-1})\!=\!\sum_{j=0}^\NC\Cflt_j(p)q^{-j}, \label{eq:MAXflt}
\end{equation} \vskip -1.5mm \noindent
where $\Bflt_i(p) = (g_i\diamond p)$ and $\Cflt_j(p) = (h_j\diamond p)$. Hence, $\Bflt_i(p)$ and $\Cflt_i(p)$ are multi-linear matrix functions in the scheduling signal and its times-shifts from $t,\ldots,t-i$, similar to the Markov coefficients~\eqref{eq:subMarkovParam}. In this paper, the LPV-MAX model~\eqref{eq:MAXflt} is identified by using the PEM setting. To avoid identifiability issues,~\eqref{eq:SSrep} with functional dependencies~\eqref{eq:sysMatrices} is assumed to be structurally state observable and state reachable w.r.t. both $u$ and $e$ in the dertiministic sense\footnote{There exists at least one $p\in\mathbb{P}^{\mathbb{Z}}$ such that $n_{\mathrm{x}}$-step observability and reachability holds for all time moments on the support of the signals~\cite{Toth2010a}.}, which implies joint minimality of \eqref{eq:SSrep}.

In order to apply the PEM framework, we select the \emph{model structure} similar to~\eqref{eq:MAXflt}: \vspace{-2mm}
\begin{equation}
y = \Bflth(\theta,p,q^{-1})u + \Cflth(\theta,p,q^{-1})\varepsilon(\theta), \label{eq:MAXmodel}
\end{equation} \vskip -1.5mm \noindent
where $\varepsilon(\theta):\sZ\rightarrow\sY$ is the one-step-ahead prediction error and the process $\Bflth$ and noise $\Cflth$ models are considered to be polynomials of $q^{-1}$ with $p$-dependent coefficients similar to $\Bflt$ and $\Cflt$ in~\eqref{eq:MAXflt} with orders $\NBh$, $\NCh$, respectively. These polynomials are parametrized in terms of the sub-Markov parameters resulting in an overall parameter vector $\theta$.

The data-generating system~\eqref{eq:SSrep} is aimed to be identified using a given identification dataset $\Dat=\{u(t),p(t),y(t)\}_{t=1}^N$ generated by~\eqref{eq:SSrep} and minimizing the following $\ell_2$ \textit{loss function} \vspace{-2mm}
\begin{equation}
V_N(\theta) = \Tr\left( \frac{1}{N} \sum_{t=1}^N \varepsilon(\theta,t) \varepsilon^{\!\top}\!\!(\theta,t)\right) \label{eq:finiteLossFunc}
\end{equation} \vskip -1.5mm \noindent
where $N$ is the number of data points collected in $\Dat$. Under weak regularity conditions it is well known that\vspace{-2mm}
\begin{equation} 
V_N(\theta)\!\rightarrow\!V_\infty(\theta) \!=\! \lim_{N\!\rightarrow\infty} \frac{1}{N}\! \sum_{t=1}^N \! \Tr\!\left( \expct\left\{ \varepsilon(\theta,t) \varepsilon^{\!\top}\!\!(\theta,t)\right\} \!\right)\!, \label{eq:infiniteLossFunc}
\end{equation} \vskip -1.5mm \noindent
with probability one and uniformly in $\theta$~\cite{Ljung1989}, which removes the effect of sample based realizations of $u$, $p$, and $e$ in the analysis. Due to the uniform convergence of~\eqref{eq:infiniteLossFunc}, the results of Section~\ref{subsec:globalMin} and~\ref{subsec:stationaryPoint} give information about the shape of $V_N(\theta)$ for large $N$.

First, let us define a notation to indicate which sub-Markov parameters~\eqref{eq:subMarkovParam} of the filters $\Bflt$, $\Cflt$ are selected. Denote with $\sI_s^v$ the set $\{s,s+1,\ldots,v\}$. Then, $[\sI_s^v]^n$ defines the set of all sequences of the form $(i_1,\ldots,i_n)$ with $i_1,\ldots,i_n\in\sI_s^v$. The elements of $\sI_s^v$ will be viewed as characters and the finite sequences of elements of $\sI_s^v$ will be referred to as strings. Then $[\sI_s^v]^n$ is the set of all strings containing exactly $n$ characters. Then a selection with $n\geq0$ is constructed from $\eta\in\left[\sI_0^\NP\right]^n_0$ with $\left[\sI_0^\NP\right]^n_0=\{\epsilon\} \cup \sI_0^\NP\cup\ldots\cup\left[\sI_0^\NP\right]^n$ and $\epsilon$ denoting the empty string. As an example, $\left[\sI_0^1\right]^2_0=\{\epsilon,0,1,00,01,10,11\}$. Define by $\numstr{\eta}$ the amount of characters of a single string in the set. With this notation, we will simplify the notation of the sub-Markov parameters as \vspace{-2mm}
\begin{equation}
C_{\left[\eta_{\mathrm{b}}\right]_1}A_{\left[\eta_{\mathrm{b}}\right]_2}\cdots A_{\left[\eta_{\mathrm{b}}\right]_{i-1}}B_{\left[\eta_{\mathrm{b}}\right]_{i}} = \Bpar_{\eta_{\mathrm{b}}}, \label{eq:MAXpar}
\end{equation} \vskip -1.5mm \noindent
where $i=\numstr{\eta_{\mathrm{b}}}$, $\left[\eta\right]_j$ denotes the $j$-th character of the string $\eta$ and $\eta_{\mathrm{b}}\in\left[\sI_0^\NP\right]^{\NB+1}_2$. Hence, $\left[\sI_0^\NP\right]^{\NB+1}_2$ is the set of all possible combinations of the sub-Markov parameters for the process filter. Similarly, we indicate a sub-Markov parameter of the noise filter  $\Cflt$, process model $\Bflth$, and noise model $\Cflth$ by $\eta_{\mathrm{c}}\in\left[\sI_0^\NP\right]^{\NC+1}_2$, $\eta_{\hat{\mathrm{b}}}\in\left[\sI_0^\NP\right]^{\NBh+1}_2$, and $\eta_{\hat{\mathrm{c}}}\in\left[\sI_0^\NP\right]^{\NCh+1}_2$, respectively. Then define\footnote{The inverse filter $\Gamfnc(\theta,p,\!q^{-1})$ of $\Cflth(\theta,p,q^{-1})$ in~\eqref{eq:dfnInvCflth} is generated by \\ $\Gamfnc(\theta,p,\!q^{-1})=\sum_{i=0}^\infty\big( I - \Cflth(\theta,p,q^{-1}) \big)^i$~\cite{Toth2012a}.} \vspace{-2mm}
\begin{subequations} \label{eq:dfnMA}
\begin{align}
Q(\theta)&=\bar\expct\left\{ \varepsilon(\theta,t) \varepsilon(\theta,t)^\top \right\} \\
\Gamfnc(\theta,p,q^{-1}) & = \sum_{i=0}^\infty \hat\Gamma_i(\theta,p,t)q^{-i} \nonumber \\ 
						   & = \Cflth^{-1}(\theta,p,q^{-1}),~(\Gamma_0\!=\!I) \label{eq:dfnInvCflth} \\
R_k(\theta,t)&=\varepsilon(\theta,t-k) \varepsilon^{\!\top}\!(\theta,t),~~k\geq0, \label{eq:dfnMA-autoCor}
\end{align} \vskip -1.5mm \noindent
\end{subequations} 
where  \vspace{-1mm}
\begin{equation*}
\bar\expct\left\{\cdot\right\} = \lim_{N\rightarrow\infty} \frac{1}{N} \sum_{t=1}^N \expct\left\{\cdot\right\}.
\end{equation*}  \vskip -1.5mm \noindent
Remark, we use a simplified notation $\hat\Gamma_i(\theta,p,t)=(\hat\Gamma_i(\theta)\diamond p)(t)$ where each $\hat\Gamma_i(\theta,p,t)$ depends on $p(t),\ldots,p(t-i)$.

\vspace{-2mm}
\subsection{Pseudo linear regression} 
The advantage of the given problem setting of Section~\ref{sec:MaxProblemSet} is that it can easly be solved by applying global pseudo linear regression methods, e.g.,~\cite[Algorithm 3]{Toth2012a}. Pseudo linear regression has a relatively low computational load. Due to space limitations, the algorithm is not presented here. If we take $\NBh$ and $\NCh$ to be finite, the corresponding LPV-MAX model can be seen as an extended LTI model. In the LTI case, pseudo linear regression has been studied extensively and it is known that it will converge if no auto-regressive part is present~\cite{Ljung1975}. However, a key part of the convergence is to have only one unique solution of the minimization problem~\eqref{eq:infiniteLossFunc}, which is analysed in Sections~\ref{subsec:globalMin} and~\ref{subsec:stationaryPoint}.

\subsection{Global minima} \label{subsec:globalMin}

To show that a global minimum exists, the following assumption is taken:
\begin{enumerate}[label=\bfseries A\arabic*,ref=A\arabic*]
 	\item\label{ass:input1} The input signal $u$ is uncorrelated to the noise signal $e$.
\end{enumerate}
With the aforementioned identification setting, we get:
\begin{thm} \label{thm:globalMin}
Given the data generating system~\eqref{eq:MAXsys} with functional dependencies~\eqref{eq:MAXflt}, model structure~\eqref{eq:MAXmodel}, $\NBh\geq\NB$, $\NCh\geq\NC$, and $u$ satisfying~\ref{ass:input1}; then $V_N(\theta)$ has a global minimum \vspace{-2mm}
\begin{equation}
\min_\theta V_\infty(\theta) = \Tr(\Sigma_{\mathrm{e}}), \label{eq-thm:minimization}
\end{equation} \vskip -1.5mm \noindent
which can be obtained if the elements of $\theta$ satisfy\vspace{-2mm}
\begin{subequations} \label{eq-thm:parmCons}
\begin{equation} 
\Bparh_{\eta_{\mathrm{b}}} = \Bpar_{\eta_{\mathrm{b}}},~~\Cparh_{\eta_{\mathrm{c}}} = \Cpar_{\eta_{\mathrm{c}}}, \label{eq-thm:parmCons1}
\end{equation}\vskip -0.5mm \noindent
for all indices $\eta_{\mathrm{b}}\!\in\!\left[\sI_0^\NP\right]^{\NB\!+\!1}_2$ and $\eta_{\mathrm{c}}\!\in\!\left[\sI_0^\NP\right]^{\NC\!+\!1}_2$. In addition, if $\NBh>\NB$, $\NCh>\NC$,\vspace{-2mm}
\begin{equation} 
\Bparh_{\eta_{\hat{\mathrm{b}}}} = 0,~~~\Cparh_{\eta_{\hat{\mathrm{c}}}} = 0, \label{eq-thm:parmCons2}
\end{equation} \vskip -0.5mm \noindent
for all $\eta_{\hat{\mathrm{b}}}\!\in\!\left[\sI_0^\NP\right]^{\NBh\!+\!1}_{\NB\!+\!2}$ and $\eta_{\hat{\mathrm{c}}}\!\in\!\left[\sI_0^\NP\right]^{\NCh\!+\!1}_{\NC\!+\!2}$. \hfill $\square$
\end{subequations}
\end{thm}
\begin{proof}
For notational ease, denote $\Bflt(p,q^{-1})$, $\Bflth(\theta,p,q^{-1})$, $\Cflt(p,q^{-1})$, $\Gamfnc(\theta,p,q^{-1})$, and $\varepsilon(\theta)$ as $\Bflt$, $\Bflth$, $\Cflt$, $\Gamfnc$, and $\varepsilon$, respectively. Then, using~\eqref{eq:MAXsys} and~\eqref{eq:MAXmodel}, let us rewrite $\varepsilon$ in terms of $u$ and $e$ as\vspace{-1mm}
\begin{equation*} 
\varepsilon = \Gamfnc\left[\left(\Bflt-\Bflth\right)u+ \Cflt e\right].
\end{equation*} \vskip -1.5mm \noindent
Hence, as $e$ and $u$ are uncorrelated (\ref{ass:input1}), the loss function~\eqref{eq:infiniteLossFunc} separates as \vspace{-2mm}
\begin{equation} \label{eq-prf:sepLoss}
V_\infty(\theta) = V_{\infty,1}(\theta)+V_{\infty,2}(\theta),
\end{equation} \vskip -1.5mm \noindent
with \vspace{-2mm}
\begin{subequations}
\begin{align}
V_{\infty,1}(\theta) &\!=\! \Tr\!\left[ \bar\expct\left\{\! \Gamfnc\!\left(\!\Bflt\!-\!\Bflth\!\right)\!u(t) u^{\!\top}\!\!(t)\!\left(\!\Bflt^{\!\top}\!\!-\Bflth^{\!\top}\!\right)\!\Gamfnc^{\!\top}\!\right\}\! \right]\!, \label{eq-prf:minProc} \\
V_{\infty,2}(\theta) &\!=\! \Tr\!\left[ \bar\expct\left\{ \Gamfnc\Cflt e(t) e^{\!\top}\!\!(t) \Cflt^\top\Gamfnc^\top\right\} \right]. \label{eq-prf:minNoise}
\end{align} \vskip -1.5mm \noindent
\end{subequations}
Obviously, $V_{\infty,1}(\theta)\geq0$ with equivalence if $\Bflt\equiv\Bflth$. For $V_{\infty,2}(\theta)$, remark that the filters $\Cflt$ and $\Cflth$ are monic, \vspace{-2mm}
\begin{equation}
\varepsilon = \Gamfnc\Cflt e = e + v,
\end{equation} \vskip -1.5mm \noindent
where the random signal $v(t)$ depends linearly on the past samples of $e$, i.e., $\{e(\tau)\}_{\tau=-\infty}^{t-1}$, but is independent of $e(t)$. Hence, it follows that \vspace{-2mm}
\begin{align}
V_{\infty,2}(\theta) &\!= \Tr\left( \bar\expct\left\{ (e(t) + v(t)) (e(t) + v(t))^\top\right\} \right) \nonumber \\
&\!= \! \Tr\left(\expct\{ e(t)e^{\!\top}\!(t)\}\right)\! +\!\! \Tr\left( \bar\expct\left\{ v(t)v^{\!\top}\!(t)\right\} \right) \nonumber \\
&\!\geq \Tr\left(\expct\{ e(t)e^{\!\top}\!(t)\}\right) = \Tr(\Sigma_{\mathrm{e}}). \label{eq-prf:globalMin}
\end{align} \vskip -1.5mm \noindent
Eq.~\eqref{eq-prf:globalMin} holds with equality if $v\equiv0$, which follows if $\Gamfnc\Cflt=I$ implying $\Cflt\equiv\Cflth$. Hence, $V_\infty(\theta)= \Tr(\Sigma_{\mathrm{e}})$ if~\eqref{eq-thm:parmCons} holds, which is a global minimum of $\min_\theta V_\infty(\theta)$.
\end{proof}

Theorem~\ref{thm:globalMin} does not imply uniqueness of the solution, only that a global minimum exists if the sub-Markov parameters of the model~\eqref{eq:MAXmodel} are equal to those of the original system~\eqref{eq:MAXflt}. Uniqueness of the solution is proven next.

\vspace{-2mm}
\subsection{Stationary point} \label{subsec:stationaryPoint}

Next, to show uniqueness of the of minimum in Theorem~\ref{thm:globalMin}, we prove that there exists only one stationary point of~\eqref{eq-thm:minimization}, hence, the global minimum of~\eqref{eq-thm:minimization} is unique and it can always be obtained by optimization methods. To prove this, the following additional assumptions are taken:
\begin{enumerate}[label=\bfseries A\arabic*,ref=A\arabic*,resume]
 	\item\label{ass:scheduling} Each signal $p_i$ is assumed to be a zero-mean white noise process with finite variance and independent of $e$. The processes $p_i$ are mutually independent and the higher order moments are bounded, given by $\expct\{p^k_ip^l_j\}=c_{i,j}^{k+j}<\infty$ for $i,j\in\sI_1^\NP$ and $2\geq k+j\geq6+\NCh$.
 	\item\label{ass:input} The given input $u$ is chosen such that the signal \vspace{-1mm}
 	\begin{equation}
 	\tilde u(t) = \left[\begin{array}{c} 1 \\ p(t-\NBh) \end{array} \right] \otimes \cdots \otimes \left[\begin{array}{c} 1 \\ p(t) \end{array} \right] \otimes u(t)
	\end{equation} \vskip -1.5mm \noindent
	is persistently exciting of order $\NY\!\!\sum_{i=1}^\NBh(1\!+\NP)^i\NU$~\cite{Ljung1989}.
\end{enumerate}
The first assumption, the choice of the scheduling signal, can be made less restrictive, however, it will make the analysis more involved. In this case, we restrict the scheduling signal to be stationary, hence, $\Tr( \bar\expct\{\cdot\})=\Tr( \expct\{\cdot\})$ for $V_{\infty,2}$ in~\eqref{eq-prf:minNoise}. The second assumption is needed to be able to uniquely distinguish all parameters in $\Bflth$, which is discussed in the proof of Theorem~\ref{thm:uniquenessMAX}.

As the loss function $V_{\infty}$ can be separated as given in~\eqref{eq-prf:sepLoss}, we will first focus on the analysis of the loss function w.r.t. the noise model, i.e., $V_{\infty,2}$~\eqref{eq-prf:minNoise}. The stationary points of~\eqref{eq-prf:minNoise} are solutions of the following equations \vspace{-1mm}
\begin{equation} \label{eq:infVder}
\frac{\partial V_{\infty,2}(\theta)}{\partial \theta_{ij,\eta_{\hat{\mathrm{c}}}}} = 0,~~1\leq i,j\leq\NY,~~\forall\eta_{\hat{\mathrm{c}}}\in\left[\sI_0^\NP\right]^{\NCh\!+\!1}_2,
\end{equation}
where $\theta_{ij,\eta_{\hat{\mathrm{c}}}}$ only contains the parameters of the noise model. Using the notation of~\eqref{eq:dfnMA}, the partial derivative~\eqref{eq:infVder} is
\begin{equation}
0=\Tr \frac{\partial\Tr Q(\theta)}{\partial Q(\theta)} \frac{\partial Q(\theta)}{\partial \theta_{ij,\eta_{\hat{\mathrm{c}}}}} = 2 \bar\expct\left\{\varepsilon^{\!\top}\!\!(\theta,t) \frac{\partial \varepsilon(\theta,t)}{\partial \theta_{ij,\eta_{\hat{\mathrm{c}}}}}\right\}. \label{eq:partialMA}
\end{equation}
Taking the partial derivative of~\eqref{eq:MAXmodel} w.r.t. the model parameters gives \vspace{-2mm}
\begin{equation}
0 = \Cflth(\theta,p,q^{-1})\frac{\partial \varepsilon(\theta,t)}{\partial \theta_{ij,\eta_{\hat{\mathrm{c}}}}}+s_is_j^\top p_{\eta_{\hat{\mathrm{c}}}}(t)\varepsilon(\theta,t-k) \label{eq:partialMAModel1}
\end{equation}
for each $t\in\sZ$, where $k=\numstr{\eta_{\hat{\mathrm{c}}}}-1$, $s_i$ is a selector vector for which only the $i$-th element is non-zero, and $p_{\eta_{\hat{\mathrm{c}}}}(t)$ defines the product of different scheduling signals and its time-shifts $p_{\eta_{\hat{\mathrm{c}}}}(t) = \prod_{i=0}^{\numstr{\eta_{\hat{\mathrm{c}}}}-1}p_{\left[\eta_{\hat{\mathrm{c}}}\right]_i}(t-i)$ and, for notational simplicity, $p_0(t)=1,~\forall t\in\sZ$. See that~\eqref{eq:partialMAModel1} is equivalent to \vspace{-1mm}
\begin{equation}
\frac{\partial \varepsilon(\theta,t)}{\partial \theta_{ij,\eta_{\hat{\mathrm{c}}}}} = -\Gamfnc(\theta,p,q^{-1})s_is_j^\top p_{\eta_{\hat{\mathrm{c}}}}(t)\varepsilon(\theta,t-k),\label{eq:partialMAModel2}
\end{equation}  \vskip -0.5mm \noindent
and substituting~\eqref{eq:partialMAModel2} in~\eqref{eq:partialMA} gives\vspace{-1mm}
\begin{align*} 
0 &= \bar\expct\!\left\{\!\varepsilon^{\!\top}\!(\theta,t)\sum_{r=0}^\infty \hat\Gamma_r(\theta,p,t)s_is_j^\top p_{\eta_{\hat{\mathrm{c}}}}(t-r)\varepsilon(\theta,t-k-r)\!\right\} \\
  &= s_j^\top \bar\expct\left\{\!\varepsilon(\theta,t-k)\sum_{r=0}^\infty\varepsilon^{\!\top}\!\!(\theta,t+r)\hat\Gamma_r(\theta,p,t)p_{\eta_{\hat{\mathrm{c}}}}(t-r)\!\right\}\!s_i,
\end{align*} \vskip -0.5mm \noindent
Hence, we find that~\eqref{eq:infVder} is equivalent to: \vspace{-1mm}
\begin{equation}
\bar\expct\!\left\{\!\sum_{r=0}^\infty \!R_{k+r}(\theta,t)\hat\Gamma_r(\theta,p,t)p_{\eta_{\hat{\mathrm{c}}}}(t-r)\!\right\} \!=\! 0,~\forall\eta_{\hat{\mathrm{c}}}\!\in\!\left[\sI_0^\NP\right]^{\NCh\!+\!1}_2\!\!, \label{eq:FindSol}
\end{equation} \vskip -0.5mm \noindent
with $k=\numstr{\eta_{\hat{\mathrm{c}}}}-1$. We would like to highlight that the analysis given next differs from the LTI case~\cite{Stoica1982} as $\varepsilon$ is not independent of $p$, e.g., $\hat\Gamma_r$ and $R_{k+r}$ are not independent in~\eqref{eq:FindSol}. Hence, this dependency makes the analysis more involved. Note that the time-shift operator $q$ is non-communicative if applied on a coefficient function, e.g., $q^{-1}\Cflt_i(p)=\Cflt_i(q^{-1}p)q^{-1}$~\cite{Toth2010}. Hence, for the following analysis, the time-instance will be indicated as $\Cflt_j(p,t-i)$ for $\Cflt_j(p(t-i),\ldots,p(t-i-j))$. First, decompose $R_{k+r}$ as:

\begin{lem} \label{lem:corrEps}
If~\ref{ass:scheduling} holds, then are the matrices $R_{r}$ for $r\geq1$ defined in~\eqref{eq:dfnMA-autoCor} given by \vspace{-2mm}
\begin{multline}\label{eq-lem:expautoCor}
\bar\expct\{ \!R_{r}(\theta,t)\! \} \!=\! \bar\expct\left\{\! e(t\!-\!r)e^{\!\top}\!(t\!-\!r) \Omega^{\!\top}\!(\theta,p,t\!-\!r)\bar{\Gamma}^\top_{r}(\theta,p,t)\! \right\}   \\
	+\bar\expct\Big\{ \sum_{i=1}^\infty \bar{\Gamma}_{i}(\theta,p,t-r)\Omega(\theta,p,t-i-r) e(t-r-i)\cdot \\
		e^{\!\top}\!(t-r-i) \Omega^{\!\top}\!(\theta,p,t-i-r) \bar{\Gamma}^\top_{i+r}(\theta,p,t) \Big\}
\end{multline} \vskip -1.5mm \noindent
with \vspace{-2mm}
\begin{subequations} \label{eq-lem:defBigMat}
\begin{align}
\bar{\Gamma}_k(\theta,p,t) &\!=\! \left[\! \begin{array}{ccc} \hat\Gamma_{k-1}(\theta,p,t)&\ldots&\hat\Gamma_{k-\NCh}(\theta,p,t) \end{array}\! \right], \\
\Omega(\theta,p,t) &\!=\! \left[\! \begin{array}{c}   \Cflt_1(p,t+1)-\Cflth_1(\theta,p,t+1) \\
										 \vdots \\ 
										 \Cflt_\NC(p,t+\NC)-\Cflth_\NC(\theta,p,t+\NC) \\
						 				-\Cflth_{\NC+1}(\theta,p,t+\NC+1) \\ 
						 				 \vdots \\
						 				-\Cflth_\NCh(\theta,p,t+\NCh)\end{array} \! \right]\!,
\end{align} \vskip -1.5mm \noindent
\end{subequations}
where $\Gamma_{l}(\cdot)=0$ for $l<0$. \hfill $\square$
\end{lem}
\begin{proof}
See Appendix.
\end{proof}
Now we return to the uniqueness of $V_{\infty,2}$:

\begin{lem} \label{lem:uniquenessMA}
Given the data generating system~\eqref{eq:MAXsys} with functional dependencies~\eqref{eq:MAXflt}, model structure~\eqref{eq:MAXmodel} with $\NCh\geq\NC$, $\Bflt(p,q^{-1})=\Bflth(\theta,p,q^{-1})\equiv0$, then, under assumption~\ref{ass:scheduling} the global minimum for \vspace{-2mm}
\begin{equation}
\min_{\theta} V_{\infty,2}(\theta) = \Tr(\Sigma_{\mathrm{e}}),
\end{equation} \vskip -1.5mm \noindent
is the unique stationary point of $V_{\infty,2}$.\hfill $\square$
\end{lem}
\begin{proof}
See that~\eqref{eq-lem:expautoCor} can be substituted in~\eqref{eq:FindSol} as $e$ is independent from $p$. Hence, \vspace{-2mm}
\begin{multline} \label{eq-prf:FindSol}
\bar\expct\Big\{\!\sum_{r=0}^\infty \Big[ e(t-r-k)e^{\!\top}\!(t-r-k) \Omega^{\!\top}\!(\theta,p,t-r-k)\bar{\Gamma}^\top_{\!r-k}(\theta,p,t) \\
	\! + \!\sum_{i=1}^\infty\! \bar{\Gamma}_{\!i}(\theta,p,t-r-k)\Omega(\theta,p,t-i-r-k) e(t-r-i)e^{\!\top}\!\!(t-i-r-k) \\
		 \Omega^{\!\top}\!\!(\theta,p,t-i-r-k) \bar{\Gamma}^{\!\top}_{\!i+r-k}(\theta,p,t)\!\Big]\!\hat\Gamma_r(\theta,p,t)p_{\eta_{\hat{\mathrm{c}}}}(t-r)\!\Big\} \!=\! 0, 
\end{multline} \vskip -2.5mm \noindent
for all $\eta_{\hat{\mathrm{c}}}\in\left[\sI_0^\NP\right]^{\NCh\!+\!1}_2$ and $k=\numstr{\eta_{\hat{\mathrm{c}}}}-1$. Remark that $\Cparh_{\eta_{\hat{\mathrm{c}}}}=0$ cannot be a solution of~\eqref{eq-prf:FindSol} as $\hat\Gamma_0=I$. Also see that the maximal amount of products of $p$ with equivalent time-shift is $6+\NCh$ in~\eqref{eq-prf:FindSol}. Therefore, in~\ref{ass:scheduling}, we have that all moments of $p$ up to the $6+\NCh$ moment are needed to be bounded. To treat~\eqref{eq-prf:FindSol} in the rest of the proof would require additional technical details, hence, instead we present just the concept of the proof w.r.t. the restricted case of $\eta_{\hat{\mathrm{c}}}\in\left[0\right]^{\NCh\!+\!1}_2$ such that $p_{\eta_{\hat{\mathrm{c}}}}(t-r)=1$, then, due to the construction of $\hat\Gamma_i$, it is impossible that summations and products of $\hat\Gamma_i$ cancel each other, which could lead to the zero solution of~\eqref{eq-prf:FindSol}. Hence, the solution can only be found for \vspace{-1mm}
\begin{equation}
\bar\expct\!\left\{\!\Omega(\theta,p,t-j) e(t-j)e^{\!\top}\!(t-j) \Omega^{\!\top}\!(\theta,p,t-j) \!\right\}\!=\!0,
\end{equation} \vskip -1.5mm \noindent
for $j>0$. We can rewrite $\Omega$ as \vspace{-2mm}
\begin{equation} \label{eq-prf:fromLPVtoLTI}
\Omega(\theta,p,t-j) = (\theta_{\Cpar}-\hat\theta_{\Cparh})P(t-j),
\end{equation} \vskip -1.5mm \noindent
with \vspace{-2mm}
\begin{align*}
\Cpar^{[n]} &\!=\! \left[ ~\Cpar_{0\ldots00} ~ \Cpar_{0\ldots01} ~ \mbox{{\small\dots}}~  \Cpar_{0\ldots0\NC} ~  \Cpar_{0\ldots10} ~ \mbox{{\small\dots}} ~ \Cpar_{\NC\ldots\NC\NC} ~\right], \\
\theta_{\Cpar} &\!=\! \left[ \begin{array}{ccccc} 
\Cpar^{[1]}   & 0 & &\ldots & 0 \\
\vdots 	      &   & & \ldots & 0 \\
\Cpar^{[\NC]} &   & 0& \ldots & 0 \\
0		      &  & & \ldots & 0
\end{array} \right],
\end{align*} \vskip -1.5mm \noindent
%
%
where $\Cpar^{[n]}\in\mathbb{R}^{\NY\times\NY(1+\NP)^{n+1}}$ denotes the matrix of all sub-Markov parameters associated with the $h_n$-th Markov coefficient. For $\theta_{\Cpar}$, the last $\NC-\NCh$ rows and columns are zero, such that $\theta_{\Cpar}$ and $\hat\theta_{\Cparh}$ have equivalent dimensions. The matrix $\hat\theta_{\Cparh}$ is similarly parametrized as $\theta_{\Cpar}$, however, with $\Cparh$, $\NCh$ in stead of $\NC$, $\Cpar$, and \vspace{-2mm}
\begin{equation*}
P(t-j) = \left[\begin{array}{c} 1 \\ p(t-j+\NCh) \end{array} \right] \otimes \cdots \otimes \left[\begin{array}{c} 1 \\ p(t-j) \end{array} \right] \otimes I_{ny},
\end{equation*} \vskip -1.5mm \noindent
where $\otimes$ is the Kronecker product. Combing gives \vspace{-1mm}
\begin{equation*}
(\theta_{\Cpar}-\hat\theta_{\Cparh})\!\underbrace{\expct\!\left\{\!P(t-j) e(t-j)e^{\!\top}\!(t-j) P^{\!\top}\!\!(t-j)\!\right\}}_{\Lambda}\! (\theta_{\Cpar}-\hat\theta_{\Cparh})^{\!\!\top}\!\!=\!0,
\end{equation*} \vskip -1.5mm \noindent
where $\Lambda$ is clearly positive definite if~\ref{ass:scheduling} holds, therefore, $\theta_{\Cpar}=\hat\theta_{\Cparh}$ is the only solution for $p_{\eta_{\hat{\mathrm{c}}}}(t-r)=1$. If $p_{\eta_{\hat{\mathrm{c}}}}(t-r)\neq1$, then the proof is more involved, but relies on the above repeated concept. Hence, it is not presented here. Concluding, $\theta_{\Cpar}=\hat\theta_{\Cparh}$ is the only solution for which the set of equations~\eqref{eq:FindSol} is satisfied. Therefore, it is the only unique stationary point of $V_{\infty,2}$.
\end{proof}

Lemma~\ref{lem:uniquenessMA} is the LPV extension of the LTI result of~\cite{Stoica1982}. Using this preliminary analysis the following result holds:
\begin{thm} \label{thm:uniquenessMAX}
Given the data generating system~\eqref{eq:MAXsys} with functional dependencies~\eqref{eq:MAXflt}, model structure~\eqref{eq:MAXmodel}, $\NBh\geq\NB$, $\NCh\geq\NC$, then under assumptions~\ref{ass:input1}-\ref{ass:input} the global minimum \vspace{-1mm}
\begin{equation}
\min_{\theta} V_{\infty}(\theta) = \Tr(\Sigma_{\mathrm{e}}),
\end{equation} \vskip -1.5mm \noindent
is the only stationary point of $V_{\infty}$ and is found if~\eqref{eq-thm:parmCons} is satisfied. \hfill $\square$
\end{thm}
\begin{proof}
Recall from~\eqref{eq-prf:sepLoss} that $V_{\infty} =V_{\infty,1}+V_{\infty,2}$. Lemma~\ref{lem:uniquenessMA} proofs that for $V_{\infty,2}$ there is only one stationary point. Hence, $\Cflt(p,q^{-1})=\Cflth(\theta,p,q^{-1})$ at this stationary point. The underlying filter $\Cflth$ is a polynomial in $q^{-1}$ with parameter-varying coefficients and it is monic; hence it is full rank in the functional sense and, therefore, its inverse $\Gamfnc$ is also full rank. Next, we need to prove that $V_{\infty,1}$ has only one stationary point. If we consider $\tilde u$ as the input signal, then the filters $\Bflt$ and $\Bflth$ can be written as a multi-input LTI filter, similar to~\eqref{eq-prf:fromLPVtoLTI}. The extended input signal $\tilde u$ is persistent of excitation of order $\NY\sum_{i=1}^\NBh(1+\NP)^i\NU$, thus $[\Bflt-\Bflth]u$ cannot be zero for $\Bflt\neq\Bflth$. Recall, $\Gamfnc(p,q^{-1})$ also cannot lose rank, hence, there exists only one stationary point of $V_{\infty,1}$ satisfying~\eqref{eq-thm:parmCons}, e.g., see~\cite{Ljung1989}. Therefore, the only stationary point of $V_{\infty}$ is found for~\eqref{eq-thm:parmCons}.
\end{proof}

Theorem~\ref{thm:uniquenessMAX} proofs that under the conditions~\ref{ass:input1}-\ref{ass:input} there can only exists one stationary point of~\eqref{eq:infiniteLossFunc}. Hence, identification of the LPV-MAX model~\eqref{eq:MAXmodel} is consistent and unbiased.

\section{LPV-SS realization} \label{sec:ho-kalman}

The next step in the proposed identification scheme is to realize the LPV-SS model from the identified LPV-MAX model. By treating the noise as an additional input, i.e., extending $ \tilde B_i \!=\! \left[\! \begin{array}{cc} B_i & K_i  \end{array} \!\right]$ for $i=0,\ldots,\NP$, an isomorphic to the original LPV-SS representation \eqref{eq:SSrep}\footnote{There exists a constant, nonsingular transformation matrix $T\in\mathbb{R}^{\NX\times\NX}$ such that: $T \hat A_i  = A_i T$, $T [\hat B_i~ \hat K_i] = [B_i~K_i]$, and $\hat C_i = C_i T$, $\forall i\in\sI_0^\NP$.} is obtained by employing the bases reduced Ho-Kalman realization scheme of~\cite{Cox2015}.
This bases reduced realization can considerably decrease the size of the Hankel matrix by selecting only its non-repetitive elements and, therefore, reducing the computational load, compared to realization on the full Hankel matrix~\cite[Eq. (48)]{Toth2012}. In the basis reduced realization, the SVD is only applied on a $\NO\times\NR$ matrix with $\NO,\NR\geq\NX$ instead of a matrix with size $\NY\sum_{l=1}^i(1+\NP)^l\times(\NU+\NY)\sum_{l=1}^j(1+\NP)^l$, for $i,j\geq2$ in the full realization case. The realization scheme allows to reconstruct the state bases of the LPV-SS model based upon only the process $\Bflt$ or noise model $\Cflt$, simply by only selecting their corresponding sub-Markov to fill the $\hank_{\selO,\selR}$ matrix in~\cite[Eq. (20)]{Cox2015}. This can be useful, if, for example, one of both models is poorly estimated then the other model can be used to more accurately reconstruct the reachability and observability matrix.

\vspace{-2mm}
\section{Simulation Example} \label{sec:example}

The proposed identification scheme is tested on the benchmark model used in~\cite{Verdult05}. The benchmark contains an MIMO LPV-SS model with input dimension $\NU=2$, scheduling dimension $\NP\!=\!2$, state dimension $\NX\!=\!2$, and output dimension $\NY\!=\!2$. We added the following parameter independent function \vspace{-2mm}
\begin{equation*}
\Kfnc = \left[ \begin{array}{cc}
0.32 & 0.16 \\ 
0.64 & 0.24
\end{array}  \right].
\end{equation*} \vskip -1.5mm \noindent
The simulation output or one-step-ahead predicted output $\hat{y}$ of the estimated model is compared to measured output or the one-step-ahead predicted output $y$ of the oracle, respectively, by means of the \textit{best fit rate} (BFR)\footnote{Usually the BFR are defined per channel. Eq.~\eqref{eq:BFR} are the average performance criteria over all channels.} \vspace{-1mm}
\begin{align}
\BFR&=\max\hspace{-1mm}\left\{\hspace{-0.5mm}1\hspace{-1mm}-\hspace{-1mm}\frac{\frac{1}{N}\hspace{-1mm}\sum_{t=1}^N\hspace{-1mm}\Vert y_t-\hat{y}_t\Vert_2}{\frac{1}{N}\hspace{-1mm}\sum_{t=1}^N\hspace{-1mm}\Vert y_t-\bar{y}\Vert_2}, 0\hspace{-0.5mm}\right\} \cdot 100\%,\label{eq:BFR}
\end{align} \vskip -1.5mm \noindent
using a validation dataset $\Dval$ as in~\cite{Verdult05}. In \eqref{eq:BFR}, $\bar{y}$ defines the mean of the simulation output or one-step-ahead predicted output $y$ of the oracle. In the realization step, the basis reduced scheme uses $\NR=10$, $\NO=8$ bases, where the controllability matrix is spanned by $\selR=\!\{\!(\epsilon,0,1),\!(\epsilon,0,2),\!(\epsilon,1,1),\!(0,0,1),\!(0,0,2),\!(1,0,1),\!(1,0,2),\\(2,0,1),\!(2,0,2),\!(1,1,1)\!\}$ and the observability is spanned by $\selO = \{ (1,0,\epsilon),\!(2,0,\epsilon),\!(1,0,0),\!(2,0,0),\!(1,0,1),\!(2,0,1),\\ (1,0,2),\!(2,0,2)\}$. The truncation order is $\NBh=4$ and $\NCh=2$. To evaluate the statistical properties of the identification scheme, $N_{\mathrm{MC}}=100$ Monte Carlo runs are carried out. In each run, a new realization of the input and scheduling signal is used. We use two different methods of finding the search direction in each of the pseudo linear regression (PLR) iterations:
\begin{enumerate*}[label=\roman*)]
	\item $\ell_2$ regularized least squares estimate, or
	\item enhanced Gaus-Newton optimization method~\cite{Wills2008}.
\end{enumerate*}
The extension is trivial and, therefore, is not given here. For the $\ell_2$ regularized least squares estimate, we optimized the regularization parameters by a line-search and found $\lambda_1=1$, $\lambda_2=100$ for $\SNR=\infty$dB and $\lambda_1=0.1$, $\lambda_2=1$ for $\SNR=\{40,10\}$dB ($\lambda_1$ is w.r.t. the process parameters and $\lambda_2$ w.r.t. the noise parameters). 

Table~\ref{tbl:total} displays the mean and standard deviation of the $\BFR$ of the identification algorithm for different $\SNR_y=\{\infty,40,10\}$dB. The table indicates that all approaches are capable of identifying the underlying dynamics. However, our approach is mildly outperformed by~\cite{Wingerden2009a}. The approache of~\cite{Wingerden2009a} has numerically efficient implementation, in terms of a kernel based approach, which can also be used in our case. In addition, we noticed that the iterations of the PLR are not very robust, as reflected by the increased standard deviation. On the other hand, with this paper, we would like to show how to conceptually reduce the amount of parameters used to identify the input-output model if $\Kfnc$, $\Cfnc$, and $\Dfnc$ are parameter varying. Hence, evolving to an numerical efficient implementation is for future research.

\begin{table*}
\caption{\footnotesize Mean and standard deviation (std) of the $\BFR$ of the identification algorithm per Monte-Carlo run for different $\SNR_y=\{\infty,40,10\}$dB. The $\BFR$ is based on the simulated output and one-step-ahead predicted output of the estimated model on the $\Dval$ for $N_{\mathrm{MC}}=100$ Monte-Carlo simulations. GN indicates the enhanced Gaus-Newton optimization method~\cite{Wills2008}. The PBSID has a past and future window equal to 3.} \label{tbl:total} \vspace{-2mm} \centering
{ \small
\begin{tabular}{|c||l|l|l|l|l|l|} \hline
			& \multicolumn{2}{c|}{LPV-MAX and Regularized LS} & \multicolumn{2}{c|}{LPV-MAX and GN} & \multicolumn{2}{c|}{~\cite{Wingerden2009a}} \\ \hline
			& Simulation & Prediction & Simulation & Prediction & Simulation & Prediction \\\hline
$\infty$dB 	& $96.46$	($0.5843$)	& $96.46$ ($0.5840$)	& $96.69$ ($0.5511$)	& $97.21$ ($0.4354$)	& $99.50$ ($0.1164$)	& $99.64$ ($0.0834$) \\
$40$dB		& $96.52$	($0.5954$)	& $97.08$ ($0.4496$)	& $96.87$ ($0.5429$)	& $97.17$ ($0.4163$)	& $99.48$ ($0.1306$)	& $99.04$ ($0.0553$) \\
$10$dB		& $90.78$	($1.570$)	& $84.98$ ($1.780$)	& $90.76$ ($1.662$)	& $84.29$ ($1.804$)	& $93.06$ ($0.8757$)	& $85.25$ ($0.5100$) \\\hline
\end{tabular}
}
\vspace{-6mm}
\end{table*}

\vspace{-1mm}
\section{Conclusions} \label{sec:conclusion}

In this paper, an identification scheme to estimate LPV-SS models in innovation form with static and affine dependence on the scheduling signal has been presented. The LPV-SS representation has been reformulated into its corresponding LPV-MAX system, which significantly reduced the amount of parameters of the LPV-IO model compared to other predictor based schemes. We had proven that the LPV-MAX model could consistently be identified using the prediction error minimization framework, under some mild conditions on the input and scheduling signals. The LPV-SS model has been realized from the LPV-MAX model by using the basis reduced Ho-Kalman realization scheme. This realization scheme significantly reduces the amount of parameters needed to realize the LPV-SS model. Hence, the overall scheme decrease the computational load significantly, especially with moderate to large LPV-SS models where $\Kfnc$, $\Cfnc$, and $\Dfnc$ are parameter varying. The proposed identification scheme has been illustrated with a simulation example.

\vspace{-3mm}
\appendix[Proof of Lemma~\ref{lem:corrEps}] \label{app:proofLemCorrEps}

First, we show how the sequence $\{\hat\Gamma_i(\cdot)\}$ is generated. Note that $\Gamfnc(\theta,p,t)\Cflth(\theta,p,t)=I$. Hence, collecting all terms in the polynomial with equivalent time-shift $q^{-i}$, it follows that for $i\geq 1$ \vspace{-2mm}
\begin{multline} \label{eq-lem:conInv}
\hat\Gamma_i(\theta,p,t)+\hat\Gamma_{i-1}(\theta,p,t)\Cflth_1(\theta,p,t-i+1)\\+\ldots+\hat\Gamma_{i-\NCh}(\theta,p,t)\Cflth_\NCh(\theta,p,t-i+\NCh) = 0,
\end{multline} \vskip -1.5mm \noindent
with $\hat\Gamma_{i}(\cdot)=0$ for $i<0$ and $\hat\Gamma_0(\cdot)=I$. For example, \vspace{-2mm}
\begin{subequations} \label{eq-prf:exmGam}
\begin{align}
\hat\Gamma_1(\theta,p,t)&=-\Cflth_1(\theta,p,t),\\
\hat\Gamma_2(\theta,p,t)&= \Cflth_1(\theta,p,t)\Cflth_1(\theta,p,t-1)-\Cflth_2(\theta,p,t).
\end{align} \vskip -1.5mm \noindent
\end{subequations}
Second, rewrite $R_r(\theta,t) = \varepsilon(\theta,t-r)\varepsilon^{\!\top}\!(\theta,t)$ as \vspace{-2mm}
\begin{multline}
\expct\left\{R_r(\theta,t)\right\} = \expct\Big\{\sum_{i=0}^\infty \hat\Gamma_i(\theta,p,t-r) \Cflt(p,t-r-i)e(t-r-i) \cdot \\
\Big[ \sum_{j=0}^\infty \hat\Gamma_j(\theta,p,t) \Cflt(p,t-j)e(t-j)\Big]^\top\Big\}.\label{eq-prf:expCorExt}
\end{multline} \vskip -1.5mm \noindent
As $p$ is uncorrelated with $e$, it is clear that all terms in~\eqref{eq-prf:expCorExt} for $r-k\neq j$ are zero. Hence, \vspace{-2mm}
\begin{multline} \label{eq-prf:expCorExtUn}
\expct\left\{R_r(\theta,t)\right\} = \expct\Big\{\sum_{i=0}^\infty  \tilde{\Gamma}_{i}(\theta,p,t-r) \tilde{\Cflt}(p,t-r-i)e(t-r-i) \cdot \\
e^{\!\top}\!(t-r-i) \tilde{\Cflt}^{\top}\!(p,t-i-r)\tilde{\Gamma}^\top_{i+r}(\theta,p,t)\Big\},
\end{multline} \vskip -1.5mm \noindent
with \vspace{-2mm}
\begin{align*} 
\tilde{\Gamma}_i(\theta,p,t) &= \left[ \hat\Gamma_{i}(\theta,p,t)~\ldots~\hat\Gamma_{i-\NC}(\theta,p,t) \right], \\
\tilde{\Cflt}(p,t) &= \left[ \Cflt^\top_0(p,t)~\ldots~\Cflt^\top_\NC(p,t+\NC) \right]^\top.
\end{align*} \vskip -1.5mm \noindent
However, using the construction of $\Gamfnc$~\eqref{eq-lem:conInv} and the fact that $\Gamfnc$ and $\Cflt$ are monic polynomials, the product between $\Gamfnc$ and $\Cflt$ can be simplified. For example, \vspace{-2mm}
\begin{multline*}
\tilde{\Gamma}_{1}(\theta,p,t)\tilde{\Cflt}^\top(p,t-2) = \Cflth_1(\theta,p,t) \Cflth_1(\theta,p,t-1)- \Cflth_2(\theta,p,t) \\
-\Cflth_1(\theta,p,t) \Cflt_1(p,t-1)+ \Cflt_2(p,t)  =  \Cflt_2(p,t) -\Cflth_2(\theta,p,t) \\ + \Gamma_1\!(\theta,p,t)\!\!\left[\!\Cflt_1(p,t-1)\!-\!\Cflth_1(\theta,p,t-1)\!\right] \!\!=\!\bar{\Gamma}_{2}(\theta,p,t)\Omega(\theta,p,t-2)
\end{multline*} \vskip -1.5mm \noindent
Hence, by using~\eqref{eq-lem:defBigMat}, the above example generalizes to \vspace{-1mm}
\begin{subequations} \label{eq-prf:subsReq}
\begin{equation}
\tilde{\Gamma}_{i+r}(\theta,p,t) \tilde{\Cflt}(p,t-i-r) = \bar{\Gamma}_{i+r}(\theta,p,t)\Omega(\theta,p,t-i-r),
\end{equation} \vskip -1.5mm \noindent
and, for $i>0$ \vspace{-1mm}
\begin{equation}
\tilde{\Gamma}_{i}(\theta,p,t-r) \tilde{\Cflt}(p,t-i-r) = \bar{\Gamma}_{i}(\theta,p,t-r)\Omega(\theta,p,t-i-r).
\end{equation} \vskip -1.5mm \noindent
\end{subequations}
Substituting~\eqref{eq-prf:subsReq} in~\eqref{eq-prf:expCorExtUn} results in~\eqref{eq-lem:expautoCor}.

\bibliographystyle{IEEETran}
\bibliography{../../../library}

\begin{thebibliography}{10}
\providecommand{\url}[1]{#1}
\csname url@samestyle\endcsname
\providecommand{\newblock}{\relax}
\providecommand{\bibinfo}[2]{#2}
\providecommand{\BIBentrySTDinterwordspacing}{\spaceskip=0pt\relax}
\providecommand{\BIBentryALTinterwordstretchfactor}{4}
\providecommand{\BIBentryALTinterwordspacing}{\spaceskip=\fontdimen2\font plus
\BIBentryALTinterwordstretchfactor\fontdimen3\font minus
  \fontdimen4\font\relax}
\providecommand{\BIBforeignlanguage}[2]{{%
\expandafter\ifx\csname l@#1\endcsname\relax
\typeout{** WARNING: IEEEtran.bst: No hyphenation pattern has been}%
\typeout{** loaded for the language `#1'. Using the pattern for}%
\typeout{** the default language instead.}%
\else
\language=\csname l@#1\endcsname
\fi
#2}}
\providecommand{\BIBdecl}{\relax}
\BIBdecl

\bibitem{Wingerden2009a}
J.~W. {van Wingerden} and M.~Verhaegen, ``{Subspace identification of bilinear
  and LPV systems for open- and closed-loop data},'' \emph{Automatica},
  vol.~45, no.~2, pp. 372--381, 2009.

\bibitem{Wills2011}
A.~Wills and B.~Ninness, ``{System identification of linear parameter varying
  state-space models},'' in \emph{{Linear parameter-varying system
  identification: new developments and trends}}, P.~{Lopes dos Santos}, T.~P.
  {Azevedo Perdico\'{u}lis}, C.~Novara, J.~A. Ramos, and D.~E. Rivera,
  Eds.\hskip 1em plus 0.5em minus 0.4em\relax World Scientific, 2011, ch.~11,
  pp. 295--316.

\bibitem{Verdult2003}
V.~Verdult, N.~Bergboer, and M.~Verhaegen, ``{Identification of fully
  parameterized linear and nonlinear state-space systems by projected gradient
  search},'' in \emph{Proc. of the 13th IFAC Symposium on System
  Identification}, Rotterdam, The Netherlands, Aug. 2003, pp. 737--742.

\bibitem{Toth2012}
R.~T\'{o}th, H.~S. Abbas, and H.~Werner, ``{On the state-space realization of
  LPV input-output models: practical approaches},'' \emph{IEEE Trans. on
  Control Systems Technology}, vol.~20, no.~1, pp. 139--153, Jan. 2012.

\bibitem{Wingerden2009}
J.~W. {van Wingerden}, I.~Houtzager, F.~Felici, and M.~Verhaegen,
  ``{Closed-loop identification of the time-varying dynamics of variable-speed
  wind turbines},'' \emph{Int. J. of Robust and Nonlinear Control}, vol.~19,
  no.~1, pp. 4--21, 2009.

\bibitem{Larimore2015}
W.~E. Larimore, P.~B. Cox, and R.~T\'{o}th, ``{{CVA} identification of
  nonlinear systems with {LPV} state-space models of affine dependence},'' in
  \emph{Proc. of the American Control Conf.}, Chicago, IL, USA, Jul. 2015, pp.
  831--837.

\bibitem{Tanelli2008}
M.~Tanelli, D.~Ardagna, and M.~Lovera, ``{On-and off-line model identification
  for power management of Web service systems},'' in \emph{Proc. of the 47th
  IEEE Conf. on Decision and Control}, Cancun, Mexico, 2008, pp. 4497--4502.

\bibitem{Luspay2009}
T.~Luspay, B.~Kulcs\'{a}r, J.~W. {Van Wingerden}, and M.~Verhaegen, ``{On the
  identification of LPV traffic flow model},'' in \emph{Proc. of the European
  Control Conf.}, Budapest, Hungary, Aug. 2009, pp. 1752--1757.

\bibitem{Mohammadpour2012}
J.~Mohammadpour and C.~Scherer, Eds., \emph{{Control of linear parameter
  varying systems with applications}}.\hskip 1em plus 0.5em minus 0.4em\relax
  Springer, 2012.

\bibitem{Nemani1995}
M.~Nemani, R.~Ravikanth, and B.~A. Bamieh, ``{Identification of linear
  parametrically varying systems},'' in \emph{Proc. of the 34th IEEE Conf. on
  Decision and Control}, New Orleans, LA, USA, 1995, pp. 2990--2995.

\bibitem{Gaspar2007}
P.~G\'{a}sp\'{a}r, Z.~Szab\'{o}, and J.~Bokor, ``{A grey-box identification of
  an LPV vehicle model for observer-based side-slip angle estimation},'' in
  \emph{Proc. of the American Control Conf.}, New York City, USA, Jul. 2007,
  pp. 2961--2965.

\bibitem{Novara2011}
C.~Novara, ``{SM Identification of State-Space LPV systems},'' in \emph{{Linear
  parameter-varying system identification: new developments and trends}},
  P.~{Lopes dos Santos}, T.~P. {Azevedo Perdico\'{u}lis}, C.~Novara, J.~A.
  Ramos, and D.~E. Rivera, Eds.\hskip 1em plus 0.5em minus 0.4em\relax World
  Scientific, 2011, ch.~11, pp. 65--93.

\bibitem{Larimore2013}
W.~E. Larimore, ``{Identification of nonlinear parameter-varying systems via
  canonical variate analysis},'' in \emph{Proc. of the American Control Conf.},
  Washington, DC, USA, Jun 2013, pp. 2247--2262.

\bibitem{Santos2007}
P.~{Lopes dos Santos}, J.~A. Ramos, and J.~L.~M. {de Carvalho},
  ``{Identification of linear parameter varying systems using an iterative
  deterministic-stochastic subspace approach},'' in \emph{Proc. of the European
  Control Conf.}, Kos, Greece, Jul 2007, pp. 4867--4873.

\bibitem{Stoica1982}
P.~Stoica and T.~S\"{o}derstr\"{o}m, ``{Uniqueness of prediction error
  estimates of multivariable moving average models},'' \emph{Automatica},
  vol.~18, no.~5, pp. 617--620, 1982.

\bibitem{Cox2015}
P.~B. Cox, R.~T\'{o}th, and M.~Petreczky, ``{Estimation of LPV-SS models with
  static dependency using correlation analysis},'' in \emph{Proc. of the 1st
  IFAC Workshop on Linear Parameter Varying Systems}, Grenoble, France, Oct.
  2015, pp. 91--96.

\bibitem{Toth2010a}
R.~T\'{o}th, \emph{{Modeling and identification of linear parameter-varying
  systems}}.\hskip 1em plus 0.5em minus 0.4em\relax Springer, 2010.

\bibitem{Ljung1989}
L.~Ljung, \emph{{System identification: theory for the user}}.\hskip 1em plus
  0.5em minus 0.4em\relax Springer, 1989.

\bibitem{Toth2012a}
R.~T\'{o}th, P.~S.~C. Heuberger, and P.~M.~J. {Van den Hof},
  ``{Prediction-error identification of LPV systems: present and beyond},'' in
  \emph{Control of Linear Parameter Varying Systems with Applications},
  J.~Mohammadpour and C.~W. Scherer, Eds.\hskip 1em plus 0.5em minus
  0.4em\relax Springer, 2012, ch.~2, pp. 27--58.

\bibitem{Ljung1975}
L.~Ljung, T.~S\"{o}derstr\"{o}m, and I.~Gustavsson, ``{Counterexamples to
  general convergence of a commonly used recursive identification method},''
  \emph{IEEE Trans. on Automatic Control}, vol.~20, no.~5, pp. 643--652, 1975.

\bibitem{Toth2010}
R.~T\'{o}th, ``{Identification problems in the context of compressive
  sensing},'' Tech. Rep., 2010.

\bibitem{Verdult05}
V.~Verdult and M.~Verhaegen, ``{Kernel methods for subspace identification of
  multivariable LPV and bilinear systems},'' \emph{Automatica}, vol.~41, pp.
  1557--1565, 2005.

\bibitem{Wills2008}
A.~Wills and B.~Ninness, ``{On gradient-based search for multivariable system
  estimates},'' \emph{IEEE Trans. on Automatic Control}, vol.~53, no.~1, pp.
  298--306, Feb. 2008.

\end{thebibliography}

\end{document}